\documentclass[12pt]{article}

\usepackage{amsmath,amssymb,amsthm,fullpage,charter,verbatim}
\usepackage{pifont,amsfonts,amsmath,amstext,amsthm,amssymb}
\usepackage[english]{babel}
\usepackage[latin1]{inputenc}
\usepackage{graphicx,color}
\usepackage{mathtools}
\usepackage{datetime}
\usepackage{epic,eepic}
\newtheorem{theorem}{Theorem}
\newtheorem{lemma}{Lemma}

\newtheorem{definition}{Definition}
\newcommand{\eps}{\epsilon}

\newcommand*{\clNP}{$ \mathrm{\bf NP} $}
\newcommand*{\cH}{\mathcal{H}}

\newcommand{\maxlin }{MAX-E3LIN2 problem}
 
\newcommand{\bs }{\backslash} 
\newcommand{\cW }{\mathcal{W}}

\title{\bf Approximation Hardness of Graphic  TSP on Cubic Graphs \\[1ex]}
\author{
Marek Karpinski\thanks{Dept. of Computer Science and the Hausdorff
    Center for Mathematics, University of Bonn.
    Supported in part by DFG grants and the Hausdorff Grant EXC59-1/2.
    Email:~\texttt{marek@cs.uni-bonn.de}}
\and
    Richard Schmied\thanks{Dept. of Computer Science, University of Bonn.
    Work supported by Hausdorff Doctoral Fellowship.
    Email:~\texttt{schmied@cs.uni-bonn.de}}
}

\date{ 
}

\begin{document}
\maketitle

\begin{abstract}

\noindent We prove explicit approximation hardness results for the 
Graphic TSP on cubic and subcubic graphs as well as the new 
inapproximability bounds for the corresponding instances of the
(1,2)-TSP. The proof technique uses  new modular constructions
of simulating gadgets for the restricted cubic and subcubic instances.
The modular constructions used in the paper could be also of independent interest.
\end{abstract}

\section{Introduction} 

We study the Traveling Salesman Problem in the shortest path metric
completion (\emph{Graphic TSP}) of cubic as well as subcubic graphs. 
These two cases played a crucial role in some recent developments
on \emph{Graphic TSP} (cf. \cite{GLS05}, \cite{BSSS11a}, \cite{BSSS11b},
\cite{MS11}, \cite{M12}). We shed some  light on their
inapproximability status and 
prove explicit 
approximation hardness bounds of 1153/1152 for the cubic Graphic
TSP and 685/684 for the subcubic case.

We design new $3$-regular gadget amplifier construction yielding
the above bounds, and establish also new inapproximability bounds
for the cubic and subcubic instances of the (1,2)-TSP of 1141/1140
and 673/672, respectively.

The inapproximability bounds for the (1,2)-TSP improve over the bounds
of 1291/1290 and 787/786 claimed in \cite{CKK02} (see also \cite{EK06}). Our 
proof method in this paper depends on improved amplifier construction
and two transparent  and direct reduction  stages, firstly proving
approximation lower bounds for the cubic and subcubic instances of the
(1,2)-TSP, and then connecting it, in a special way, to the cubic and
subcubic instances of the Graphic TSP. 

We call an instance of (1,2)-TSP cubic and subcubic if the graph induced
by the all weight 1 edges is cubic and subcubic, respectively.

\section{Organization of the Paper}
The paper is organized as follows.
In Section~\ref{sec:prelim}, we give some basic definitions.
In Section~\ref{sec:approx}, we review some connections between the approximability 
of the Graphic TSP and the (1,2)-TSP. 
 In Section~\ref{sec:mainresults},
we formulate our main results, whereas in Section~\ref{sec:technique}, 
we describe the techniques
used in our proofs. In Section~\ref{sec:hybrid}, we introduce a
bounded occurrence Constraint Satisfaction problem called the Hybrid problem.   
In Section~\ref{sec:12tspmax5}, we describe the reduction given in \cite{KS13} 
from the Hybrid problem to the (1,2)-TSP. In Section~\ref{sec:12tspsubcub} and
 \ref{sec:12tspcub}, we introduce modular gadget constructions and
 prove explicit approximation hardness bounds 
for the (1,2)-TSP in subcubic and cubic graphs, respectively.
The inapproximability results   for the Graphic TSP in cubic and subcubic
graphs are given in Section~\ref{sec:graphtsp}.
In Section~\ref{sec:summary},  we summarize our results.

\section{Preliminaries}
\label{sec:prelim}
Given an arbitrary connected undirected graph $G=(V,E)$, 
we consider the shortest path metric completion $G'$ and
define the Graphic TSP problem for $G$ as the standard TSP
on the metric instance $G'$. Equivalently, the Graphic TSP
is the problem of finding a smallest Eulerian spanning 
multi-subgraph of $G$. We are interested here in special cases of the 
above problem for cubic (3-regular) and subcubic (maximum degree  3).
Both cases are known to be \clNP-hard in exact setting, as the Hamiltonian 
cycle problem is \clNP-hard for the $3$-regular graphs (cf. \cite{GJT76}), it can
be reduced to both (1,2)-TSP and Graphic TSP on cubic graphs. 

In order to describe a (1,2)-TSP instance, it is sufficient to specify
the edges of weight one. By constructing a graph $G=(V,E)$, the distance of
the vertices  $u$ and $v$ is defined to be 1 if $\{u, v\}\in E$
and 2 otherwise. To compute the cost of a tour, it is enough to consider
the parts of the tour traversing edges of $G$. 
 We call a vertex, in which the tour leaves or enters $G$ an \emph{endpoint}.
In addition,  
a vertex  with the property that the tour both enters and leaves $G$
 in that vertex  is called  \emph{double endpoint}, and we  count it as two endpoints.
If $n$ is the number of vertices  and $2\cdot p$ is the total number of endpoints, the 
cost of the tour is $n+p$ since every edge of weight two corresponds to two endpoints. 
On the other hand, every 
tour with cost $n + p$ has exactly $2\cdot p$ endpoints.

\section{Approximability}
\label{sec:approx}
The Graphic TSP for cubic and subcubic graphs is of special interest 
because of its connection to the famous $4/3$ conjecture on 
the integrality gap of the metric TSP (cf. \cite{BSSS11a}, \cite{BSSS11b}).  
Recently, the first polynomial time approximation algorithms with approximation
factor $4/3$ for the above problem on cubic and subcubic graphs were designed 
\cite{BSSS11b} and \cite{MS11}. This was slightly  improved beyond 4/3 bound for the case 
of $2$-connected cubic graphs \cite{CLS12}.

There was also a remarkable progress on general Graphic TSP 
(\cite{OSS11}, \cite{MS11}, \cite{M12})
leading to  the approximation factor $7/5$, cf. Seb\"o and Vygen~\cite{SV12}.

The (1,2)-TSP can be viewed as a special case of the Graphic TSP. To see this,
we simply augment the subgraph induced by all weight 1 edges in an instance of the 
(1,2)-TSP by a new vertex $z$ and add all edges connecting the original vertices
with that vertex $z$. Thus, the explicit approximation lower bound of 
$ 535/534 $ for general (1,2)-TSP is also the inapproximabiltity bound for
the general Graphic TSP. It is also known that the factor $3/2$ of Christofides'
algorithm~\cite{C76} for the general metric TSP is tight for the Graphic TSP
on cubic graphs. The best up to now approximation factor for (1,2)-TSP is 
$8/7$~\cite{BK06} (see also \cite{PY93}).  

In this paper, we attack both cubic and subcubic (1,2)-TSP and Graphic TSP,
and will use inherent connections between that problems.

\section{Main Results}
\label{sec:mainresults}
We prove the following explicit inapproximability results.

\begin{theorem}
\label{thm:12subcubic}
The 
Subcubic (1,2)-TSP is \clNP-hard to approximate to within any factor less
than $673/672$.
\end{theorem}

\begin{theorem}
\label{thm:12cubic}
The 
Cubic (1,2)-TSP is \clNP-hard to approximate to within any factor less
than $1141/1140$.
\end{theorem}

\noindent
For subcubic and cubic instances of the Graphic TSP, we prove the 
following.

\begin{theorem}
\label{thm:gr:cubic}
The 
Graphic TSP on subcubic graphs is \clNP-hard to approximate to within any factor less
than $685/684$.
\end{theorem}

\begin{theorem}
\label{thm:gr:subcubic}
The 
Graphic TSP on cubic graphs is \clNP-hard to approximate to within any factor less
than $1153/1152$.
\end{theorem}

\section{Techniques Used}
\label{sec:technique}
The method and techniques of the paper use new modular constructions
of simulating gadgets and are built upon some ideas of \cite{KS12}
and \cite{KS13}. The underlying constructions and their correctness
arguments
are presented in the subsequent sections.

\section{Hybrid Problem}
\label{sec:hybrid}
We start with defining the following  Hybrid problem (cf.~\cite{BK99}, see also  
and \cite{BK03}).

\begin{definition}[Hybrid problem] 
Given a system of linear
equations mod 2 containing n variables, $m_2$ equations with exactly two variables, and $m_3$
equations with exactly three variables, find an assignment to the variables that satisfies as
many equations as possible.
\end{definition}
\noindent
The following result is due to Berman and Karpinski~\cite{BK99}.
%
%
\begin{theorem}[\cite{BK99}]\label{ssphybridsatz}
For every constant $\epsilon \in (0,1/2)$ and $b\in \{0,1\}$, there exists instances of the 
Hybrid problem $I_\cH$
 with $42\nu$
variables, $60\nu$ equations with exactly two variables, and $2\nu$ equations 
of the form $x \oplus y \oplus z = b$ such that:
$(i)$ Each variable occurs exactly three times.
$(ii)$ It is \clNP-hard to decide whether
 there is an assignment to the variables that  
leaves at most $\epsilon\cdot \nu$ equations  unsatisfied,
or else every assignment  leaves at least $(1-\epsilon)\nu$ equations
unsatisfied.
$(iv)$ An  assignment to the variables in $I_\cH$ can be transformed 
in polynomial time into an  assignment satisfying all $60\nu$ equations
with two variables without decreasing the total number of  
 satisfied equations in $I_\cH$.
\end{theorem}
\noindent
The instances of the Hybrid problem produced in Theorem~\ref{ssphybridsatz}
 have an even more special structure, which we are going to describe. 
For this, we are going to introduce the  \maxlin:
 Given a 
systems $I$  of linear equations mod $2$ with exactly $3$ variables in each equation,
find an assignment that maximizes the number of satisfied equations in $I$.\\
For the \maxlin,  H{\aa}stad~\cite{H01} gave an  optimal 
inapproximability result stated below. 
\begin{theorem}[H\aa stad~\cite{H01}]
\label{thm:max-e3-lin} 
For every  $\eps \in (0,1/2)$, there exists a constant $B_\eps$
and  instances of the \maxlin~with
 $2\cdot \nu$ equations  such that:\\ 
$(i)$
 Each
variable in the instance occurs at most $B_\eps$ number of times.\\ 
$(ii)$ It is \clNP-hard to distinguish whether
 there is an assignment satisfying
all but at most $\eps \cdot \nu$ equations, or every assignment leaves at least $(1-\eps)\nu$ equations
unsatisfied. 
\end{theorem}
\noindent
In the following,  we describe briefly the reduction given in~\cite{BK99} from the 
\maxlin~to the Hybrid problem and give the proof of Theorem~\ref{ssphybridsatz}. 
For this, let us first recall some definitions (see also [BK03]).

Let  $G$ be a  graph and $X \subset V (G)$. We say 
that $G$ is a    $d$-regular amplifier for $X$ if the following two conditions hold:
\begin{itemize}
\item[$(i)$]
 All vertices in  $X$ have degree $(d -1)$ and all vertices in 
$V (G)\bs X$ have degree $d$.
\item[$(ii)$]
For every non-empty subset $U \subset V (G)$, we have the condition that 
$$|E(U, V (G)\bs U)| \geq  \min\{ |U \cap X|, |(V (G)\bs U) \cap X| \},$$
 where $E(U, V (G)\bs U)= \{ e\in E(G) \mid |U\cap e| =1\}$. 
\end{itemize}
We call $X$  the set of contact vertices and  $V (G)\bs X$ the set of checker
vertices. Amplifier graphs are used for proving hardness of approximation
 for Constraint Satisfaction problems, in which every
variable occurs a bounded number of times. Berman and Karpinski~\cite{BK99} 
gave a probabilistic argument on the existence of 
$3$-regular amplifiers. In particular, they constructed 
a very special amplifier graph, which they called  wheel amplifier.

A \emph{wheel amplifier} $\cW$ with $2n$ contact vertices is constructed as follows. 
We first create a Hamiltonian
cycle on $14n$ vertices with edge set $C(\cW)$. Then, we number 
the vertices $1, 2,. . . , 14n$ and select uniformly at random
a perfect matching $M(\cW)$ of the vertices whose number is not a multiple of $7$. 
The 
vertices in the matching are our checker vertices and the remaining vertices
are  our contacts. 
The set $M(\cW) \cup C(\cW)$ defines the edge set of $\cW$.
It is not hard  to see that the
degree requirements are satisfied.
Berman and Karpinski~\cite{BK99} gave a probabilistic argument to prove that with high
probability the above construction indeed produces a $3$-regular amplifier graph.

\begin{theorem}[Berman and Karpinski~\cite{BK99}]
\label{thm:wheelamp}
With high probability, wheel amplifier are $3$-regular amplifier.
\end{theorem} 
Let us proceed and give the proof of Theorem~\ref{ssphybridsatz}.
\begin{proof}[Proof of Theorem~\ref{ssphybridsatz}.]
Let  $\eps \in (0,1/2)$ be a constant  and $I$ an instance of the 
\maxlin, in which the number of occurences of each variable is bounded by $B_\eps$.
For a fixed $b\in \{0,1\}$, we can negate 
some of the  literals such that
 all equations in the instance $I$ are 
of the form $x\oplus y \oplus z= b$, where $x, y, z $ are variables 
or negated variables.

For a  variable $x_i$ in $I$, we denote by $d_i$ the number of occurences
of $x_i$ in $I_1$. For each $x_i$, 
we  create  a set of $7\cdot d_i=  \alpha$ new variables 
$Var(i)=\{ x^{i}_j  \}^\alpha_{j=1}$.
In addition, we construct a  wheel amplifier $\cW_i$ on $ \alpha$ vertices
with $d_i$ contacts.  Since $d_i$ is bounded by  a
constant, it can be accomplished in constant time.  In the remainder, we
refer to \emph{contact} and \emph{checker} variables as
$x^i_l\in Var(i)$, whose corresponding index $l$ is 
a contact and checker vertex in $\cW_i$,
respectively.  

Let us now define the equations of the new instance $I_\cH$
of the Hybrid problem.  
For each edge $\{j,k\} \in M(\cW_i)$, 
we create   $x^{i}_j \oplus 
x^{i}_k = 0$ and refer to equations of this form  as \emph{matching equations}. 
On the other hand, for each edge $\{l,t\}$  in   $C(\cW_i)$,
 we introduce  $x^{i}_l \oplus x^{i}_t = 0$.
Equations of the form $x_{i}\oplus x_{i+1}= 0$ with $i\in \{2,\ldots, \alpha -1 \}$
and $x_1 \oplus x_\alpha =0$
are  called \emph{cycle equations}, whereas  $x_{1}\oplus x_{2}= 0$
is the cycle border equation.
Finally, we replace the $j$-th  occurrence of $x_i$ 
in $I$ by the contact variable $x^{i}_\lambda$, where  $\lambda =7 \cdot j $.
Accordingly, we  have  $2\nu $ equations with three variables in $I_\cH$,
$60 \nu $ equations with two variables and
each variable appears in exactly $3$ equations.

We call an assignment to $Var(i)$ as \emph{consistent} if for $b_i\in \{0,1\}$,
we have that  $x^{i}_j=b_i$  for all $j\in [\alpha]$. 
A consistent assignment to the variables of $I_\cH$ 
is an assignment that is consistent for each  $Var(i)$.
 By using standard
arguments and  the amplifier constructed in Theorem~\ref{thm:wheelamp}, 
we are able to transform an assignment to the variables of $I_\cH$ 
into a  consistent one 
without decreasing the number of satisfied equations and the proof
of Theorem~\ref{ssphybridsatz} follows. 
\end{proof}

\section{(1,2)-TSP in  Graphs with Maximum Degree 5}
\label{sec:12tspmax5}
In this section, we describe the reduction constructed in  \cite{KS13} 
from the Hybrid problem to the (1,2)-TSP. 
In particular, this construction can be used to prove the following
 theorem.

\begin{theorem}[\cite{KS13}]
\label{thm:12gen:ks12}
The  (1,2)-TSP  is \clNP-hard to approximate to within any factor less
than $535/534$.
\end{theorem}

%
 %

\subsection{The Construction of $G^{12}_\cH$}
\label{sec:12oldconstr}
In the following, we describe briefly the reduction from the Hybrid problem 
to the (1,2)-TSP and refer for more details to  \cite{KS13} and \cite{KS12}.\\
\\  
Let $I_\cH$ be an instance of the Hybrid problem with $n$ wheels,
$60 \nu $ equations with two variables and $2\nu $ equations with two 
variables. In order to simulate the variables of $I_\cH$, we
introduce for each variable $x^l_i$ the corresponding parity 
gadget $P^l_i$ displayed 
in Figure~\ref{fig:12tsp:1}~(a). If we start in $v^l_{i,l1}$ or $v^l_{i,l0}$,
there are two ways to traverse this 
gadget visiting every vertex  only once. In the following, we
refer to those traversals as \emph{$0/1$-traversals},
which are defined in Figure~\ref{fig:12tsp:1} (b) and (c). 
\begin{figure}[h]
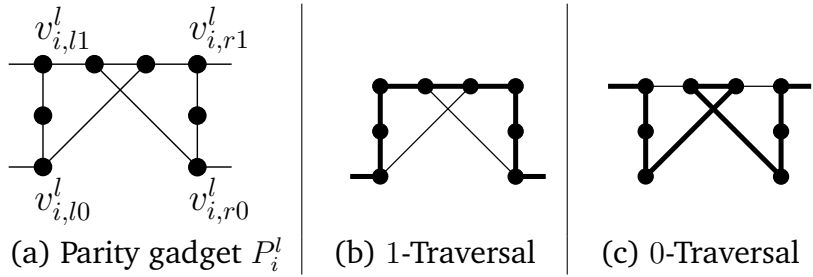

\begin{center}
\begin{tabular}[c]{c|c|c}
\input{paritydet.pdf_t} & \input{fig12tsp1a.pspdftex} & \input{fig12tsp1b.pspdftex} \\
(a) Parity gadget $P^l_i$ & (b) $1$-Traversal    & (c) $0$-Traversal 
\end{tabular}
\vspace{-0.5cm}
\end{center}
\caption{$0/1$-Traversals of the parity gadget $P^l_i$. 
Traversed edges are displayed by thick lines.}
\label{fig:12tsp:1}
\end{figure} 

\noindent
The idea of the parity gadgets is that any tour in the instance of the (1,2)-TSP
can be transformed into a tour, which uses only $0/1$-traversals of all parity gadgets
that are contained as a subgraph in $G^{12}_\cH$ without increasing its cost.
The $0/1$-traversal of the parity gadget defines the value that we assign to the 
variable  associated with the parity gadget.

For each equation, we have a specific way to connect the parity gadgets that 
are simulating the variables of the underlying equation. Let us start with
the construction for  matching equations.\\
\\
\textbf{Matching equations:} 
Given a matching equation $x^l_i \oplus x^l_j =0$ in $I_\cH$ with $i<j$ 
and the   
cycle equations
$x^l_i \oplus x^l_{i+1} =0$ and $x^l_j \oplus x^l_{j+1} =0$, we connect the associated 
parity gadgets $P^l_{i}$, $P^l_{i+1}$, $P^l_{\{i,j\}}$, $P^l_{j}$ and $P^l_{j+1}$
 as displayed in Figure~\ref{fig:12tsp:3}.
\begin{figure}[h]
\begin{center}
\input{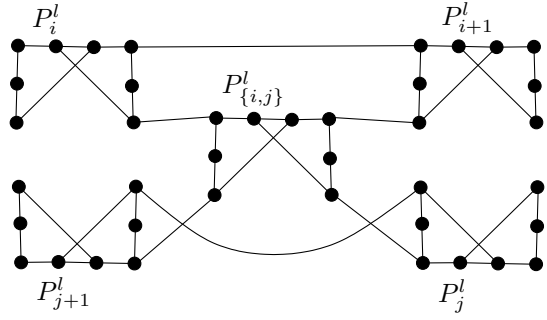}
\end{center}
\caption{Construction simulating the  equations $x^l_i \oplus x^l_j =0$,  
$x^l_i \oplus x^l_{i+1} =0$ and $x^l_j \oplus x^l_{j+1} =0$. }
\label{fig:12tsp:3}
\end{figure} 

\noindent
\textbf{Equations with three variables:}
For equations with three variables $ x\oplus y \oplus z=0=b^3_c$ 
in $I_\cH$,  we use the graph $G^{3\oplus}_c$ displayed in 
Figure~\ref{fig:12tsp:2}.  Engebretsen and Karpinski~\cite{EK06} introduced this graph 
and  proved the following statement.
\begin{lemma}[\cite{EK06}]
 \label{prop:ek06gadgettsp}
There is a simple path from    $s_c$ to  $s_{c+1}$ in Figure~\ref{fig:12tsp:2} 
 containing the vertices 
$v^1_c$ and  $v^2_c$  if and
only if an even number of parity gadgets is  traversed.
\end{lemma}
\noindent
We now explain how we connect the parity gadgets for $x_i$ and $x_{i+1}$ with 
$G^{3\oplus}_c$: 
Let us assume that  $x_i \oplus y \oplus z = 0$ and $x_{i}  \oplus x_{i+1}  =0$
are equations in $I_\cH$. We denote the parity gadgets that appear in $G^{3\oplus}_c$
as $P_{(x,i)}$, $P_y$ and $P_z$. 

Then, we connect $P^l_i$ and $P^l_{i+1}$ 
with $P_{(x,i)}$ via edges $\{ v^l_{i, r0}, v^l_{(x,i), r1}   \}$, 
$\{ v^l_{i+1, l0}, v^l_{(x,i), l1}   \}$. Furthermore, we add
$\{v^l_{i, r1}, v^l_{i+1, l1} \}$ to connect $P^l_i$ and $P^l_{i+1}$.
If $x_i$ appears negated in the equation
with three variables, we create  $\{ v^l_{i, r1}, v^l_{(x,i), r1}   \}$ and
$\{ v^l_{i+1, l1}, v^l_{(x,i), l1}   \}$ and 
$\{v^l_{i, r0}, v^l_{i+1, l0} \}$.

\begin{figure}[h]
\begin{center}
\input{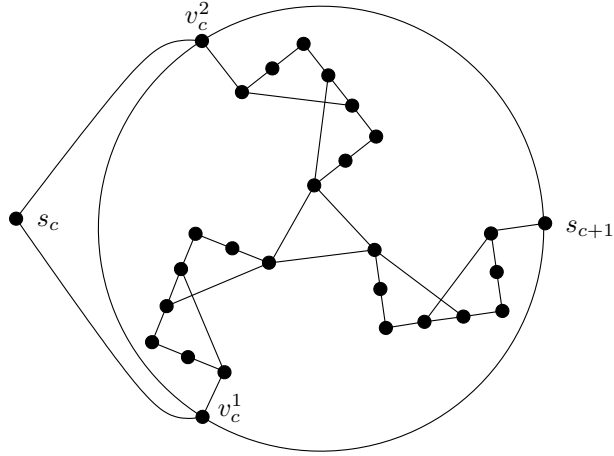}
\end{center}
\caption{Graph $G^{3\oplus }_c$ simulating  $x\oplus y \oplus z=0$.}
\label{fig:12tsp:2}
\end{figure} 

\noindent
\textbf{Cycle border equations:}
For each wheel $\cW_l$ with $l\in [n]$, we 
introduce 
three vertices  $b^1_l$, $b^2_l$ and $b^3_l$, which are connected  via $b^1_l-b^2_l-b^3_l$.
 Let  $\{ x^l_i\}^\alpha_{i=1}$
be the associated set of variables of $\cW_l$. 
Then, we connect $b^3_l$ with $v^l_{1,l0}$ and $v^l_{2,r1}$. In addition, we add
  $\{b^{1}_{l+1},  v^l_{1,l1} \}$ and $\{b^{1}_{l+1},  v^l_{2,r0} \}$.\\ 
For the last wheel, we introduce the path $b^1_{n+1}-b^2_{n+1} - s_1 $, where
$s_1$ is the starting vertex of the graph $G^{3\oplus }_1$ simulating an 
equation with three variables. The graphs corresponding  to 
equations with three variables are connected via vertices $s_1, \ldots s_{2\nu +1}$,
where $s_{2\nu +1} =  b^1_1$ is the first vertex of the path $b^1_1-b^2_1-b^3_1$.  
This is the whole description of the corresponding graph $G^{12}_{\cH}$.
\subsection{Assignment to Tour}

We are going to prove one direction of the reduction and prove the 
following lemma.
\begin{lemma}
\label{lem:gen12:easy}
Let $I_\cH$ be an instance of the Hybrid problem with $n$ wheels,
$60\nu $ equations with two variables and $2\nu $ equations with 
three variables and $\phi$ an assignment that leaves $\delta \nu   $
equations in $I_\cH$ unsatisfied for a constant $\delta \in  (0,1)$. 
Then, there is a tour in $G^{12}_\cH$
with cost at most $534\nu + 3(n + 1)-1 + \delta \nu $.
\end{lemma}
\begin{proof}
According to Theorem~\ref{ssphybridsatz}, we may assume that all variables
 associated to a wheel take the same value under $\phi$. 
 Our tour in $G^{12}_\cH$ starts in $b^1_1$ and traverses $b^1_1-b^2_1- b^3_1 $.
Then, we use the $\phi(x^1_1)$-traversals of the parity gadgets corresponding
to the variables of the wheel $\cW_1$ until  we enter the vertex $b^1_2$.
For each wheel, we use the corresponding traversal defined by the assignment.
Finally, we get to the vertex $s_1$, which belongs to the graph $G^{3\oplus}_1$.
We refer to this part of the tour as the \emph{inner loop}.
In the remaining part of the tour, we are going to traverse the graphs corresponding
to equations with three variables. If an odd number of parity gadgets
was visited in the inner loop, we can find a Hamiltonian path in $G^{3\oplus}_c$.
In the other case, we have to introduce two endpoints.
In the \emph{outer loop} of the tour, we visit all gadgets corresponding
to equations with three variables. Accordingly, our tour has cost at most 
$8\cdot 60\nu + (3\cdot 8 + 3) \cdot 2\nu + 3(n+1) -1+\delta \nu $.
\end{proof}
\subsection{Tour to Assignment}
In the following, we briefly describe the other part  of the reduction
given in \cite{KS13}.\\
\\
Let us first  introduce the notion of consistent tours. 
We call a (1,2)-tour $\pi$ in $G^{12}_\cH$ \emph{consistent} if all
parity gadgets in $G^{12}_\cH$ are visited by $\pi $ using a $0/1$-traversal.  
In order to ensure that we can restrict ourselves to consistent (1,2)-tours, 
the following  statement  can be proved. 

\begin{lemma}[\cite{KS13}]
\label{lem:12gen:cons}
Let $\pi$ be a tour in $G^{12}_\cH$. For every parity gadget $P$ in $G^{12}_\cH$,
it is possible to convert efficiently $\pi$ into a tour $\sigma$ in $G^{12}_\cH$, that
uses a $0/1$-traversal of $P$, without increasing the cost.
\end{lemma}

Due to the following lemma, we can construct efficiently an assignment if we are 
given a consistent tour in $G^{12}_\cH$. 

\begin{lemma}[\cite{KS13}]
\label{lem:12gen:hard}
Let $\pi$ be a consistent tour in $G^{12}_\cH$ with cost $534\nu + 3(n + 1)-1+\delta \nu $
for some constant $\delta \in (0,1) $. Then, it is
possible to construct efficiently an assignment  
that leaves at most $\delta \nu$ equations in $I_\cH$
unsatisfied.
 \end{lemma}
\indent
We are ready to give the proof of Theorem~\ref{thm:12gen:ks12}. 
\begin{proof}[Proof of Theorem~\ref{thm:12gen:ks12}]
Let  $I$ be an instance of the MAX-E3LIN2 problem with $n$ variables
and $2\gamma $ equations. 
For all $\tau > 0$, there exists a constant $k$ such that if we create $k$
copies of  each equation, we get an instance $I^k$
 with $2\nu  = k\cdot 2\gamma$ equations and $n$ variables with $3(n + 1)+1 \leq \nu \cdot \tau $.
From $I^k$, we generate an instance $I_\cH$ of the Hybrid problem 
 consisting of $n$ wheels,
 $60 \nu $ equations with two variables
and $2\nu $ equations with three variables. Finally, we
 construct  the associated instance $G^{12}_{\cH}$ of the
(1,2)-TSP.

Given an assignment $\phi$ to the variables of $I_\cH$ leaving $\delta \cdot \nu$ equations 
unsatisfied with $\delta \in (0, 1)$, 
according to Lemma~\ref{lem:gen12:easy},  
there is a tour 
with  length at most $534\nu 
 + 3(n + 1)-1+ \delta \cdot \nu$.

On the other hand, if we are given a tour $\sigma$ in $G^{12}_{\cH}$ with cost
$534\nu + 3(n + 1)-1+ \delta \cdot  \nu $, it is 
possible to transform $\sigma$ in polynomial time into a consistent tour $\pi$ 
without increasing the cost by applying Lemma~\ref{lem:12gen:cons} to each parity gadget
in $G^{12}_\cH$. Moreover, due to Lemma~\ref{lem:12gen:hard},   
we are able to construct efficiently an assignment, which leaves
at most $\delta \nu $ equations in $I_\cH$ unsatisfied. 

 According to Theorem~\ref{ssphybridsatz}, we know that for all $\eps> 0$, it 
is \clNP-hard to decide 
 whether there is a tour with cost
at most $ 534 \nu+ 3(n + 1)-1+ \eps \cdot \nu
 \leq 534\cdot \nu +\eps'\nu$
or all tours have cost at least $534\cdot \nu + (1 -\eps) \nu+3(n + 1)-1  \geq 535 \cdot 
\nu
-\eps'\cdot \nu  $, for some $\eps'$ which depends only on $\eps$ and $\tau$.
By appropriate choices for $\epsilon$ and $ \tau$, 
the ratio between these two cases can get arbitrarily close to $535/534$ 
\end{proof}


%
%
%
%

\section{ (1,2)-TSP in Subcubic Graphs}
\label{sec:12tspsubcub}
In this section, we are going to define a new outer loop of the construction from 
the previous section in order to obtain an instance of the 
(1,2)-TSP in subcubic graphs.\\
\\
The gadgets simulating equation with 
three variables in the  construction given in \cite{KS13} contain vertices 
with degree 5. We are going to replace these gadgets 
 by cubic graphs which we will specify later on. 
In order to understand the cubic gadgets, we first 
describe a reduction from
the \maxlin~to the MAX-2in3SAT problem. The reduction 
is straightforward: Given an equation of the form $x\oplus y \oplus z=0$,
we create  three clauses  $(x \vee a_1 \vee a_2 )$, $(y \vee a_2 \vee a_3 )$
and $(z \vee a_1 \vee a_3 )$. Note that if we are given an assignment to $x,y$ and 
$z$ that satisfies the equation mod 2, then, it is possible to find an assignment 
to $a_1$, $a_2$ and $a_3$ that satisfies  all three corresponding clauses.
In the other case, we find assignments to $a_1$, $a_2$ and $a_3$ that make 
at most two clauses satisfied.

\begin{figure}[h]
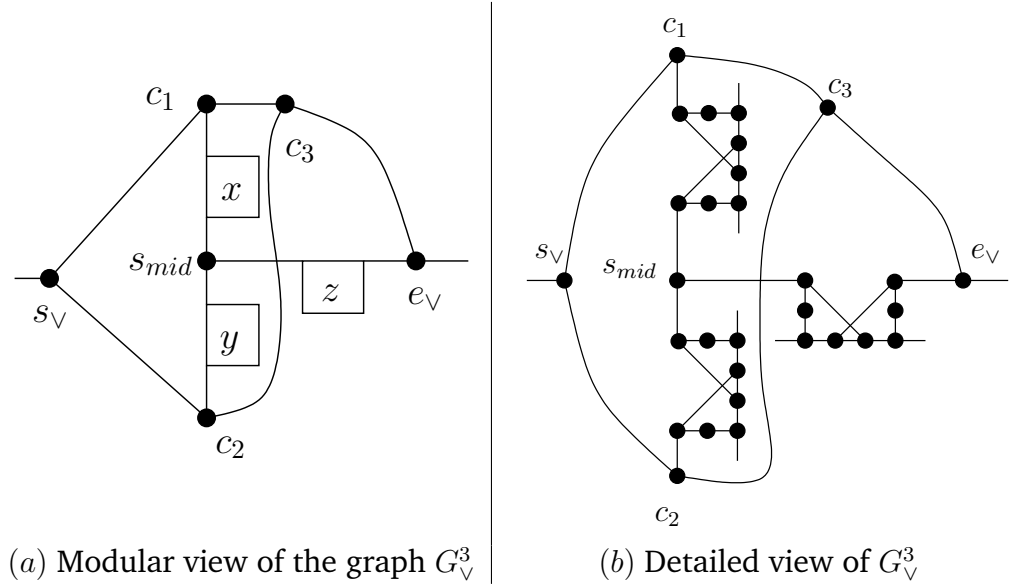

\begin{center}
\begin{tabular}[c]{c|c}
\input{max2in3.pdf_t} & ~ \input{max2in3det.pdf_t}  \\
$(a)$~Modular view of the graph $G^3_\vee$  & $(b)$~Detailed view of $G^3_\vee$ 
\end{tabular}
\end{center}
\caption{The graph $G^3_\vee$ corresponding to $(x\vee y \vee z)$.}
\label{fig:3}
\end{figure} 

In the next step, we are going to design a gadget  that simulates the 
predicate 2in3SAT. This gadget is displayed in Figure~\ref{fig:3} (a).
The boxes can be viewed as modules,  which will be replaced with 
a parity gadget or a graph with similar properties (see Figure~\ref{fig:2}). 
Any graph with less vertices and the properties of a parity gadget will 
lead to improved inapproximability factors for the corresponding problems.
Note that the graph in Figure~\ref{fig:3} (b) has degree at most $3$. 

We are going to 
prove the following lemma.
\begin{lemma}
There is a Hamiltonian path from $s_\vee$ to $e_\vee$ in the graph displayed in 
Figure~\ref{fig:3} (a) if and only if 2 edges with modules are traversed.
\end{lemma}
\begin{proof}
There are three possibilities to enter the vertex $s_{mid}$. Therefore, 
a Hamiltonian path in $G^3_\vee$ contains $(i)$~$c_1 -  s_{mid} - c_2 $, 
 $(ii)$~$c_1 -  s_{mid} - e_\vee $ or  $(iii)$~$c_2-   s_{mid} -  e_\vee $.
 In the case $(i)$, we are forced to use $ \{c_3,e_\vee\}$ and 
then, either $\{s_\vee,  c_1 \}$ and  $ \{c_3,c_2\}$
or $\{s_\vee,  c_2 \}$ and $ \{c_3,c_1\}$. In the case $(ii)$, we first note that we cannot use 
$\{e_\vee,c_3\}$.
Due to the degree condition, we have to  use $c_2- c_3 - c_1$.
The only remaining vertex with degree one is $c_2$ and has to be connected to $s_\vee$.
In case $(iii)$, we may argue similarly to case $(ii)$.
\end{proof}
As for the next step, we introduce a gadget that simulates $a^1_1 \oplus a^2_1 =0$
displayed in Figure~\ref{fig:4}. 
We see that in order to get from the vertex $s_1$ to $e_1$,
we simply use the edge $\{s_1, e_1  \} $ or the three edges which are connecting
the two parity gadgets. 
 
\begin{figure}[h]
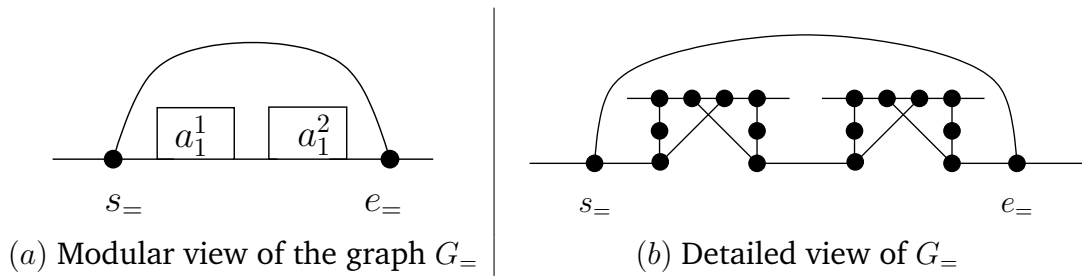

\begin{center}
\begin{tabular}[c]{c|c}
\input{equal.pdf_t} & ~ \input{equaldet.pdf_t}  \\
$(a)$~Modular view of the graph $G_{=}$  & $(b)$~Detailed view of $G_{=}$ 
\end{tabular}
\end{center}
\vspace{-0.5cm}
\caption{Graph $G_{=}$ corresponding to $a^1_1 \oplus a^2_1 =0 $.}
\label{fig:4}
\end{figure} 
\noindent
We are ready to  describe the construction that simulates the equation  
$x\oplus y \oplus z =0 $:
We create three copies of the gadget $G^3_\vee$, denoted by $G^{31}_\vee$, $G^{32}_\vee$
 and $G^{33}_\vee$, to simulate $(x \vee a^{1}_1 \vee a^{1}_2)$,
$(y \vee a^{2}_2 \vee a^{1}_3)$ and $(z \vee a^{2}_1 \vee a^{2}_3)$. For each $i\in [3]$,
the vertex set of $G^{3i}_\vee$ is defined by 
$\{s^i_\vee, c^i_1,  c^i_1, c^i_2, c^i_3 , e^i_\vee, s^i_{mid} \}$. 
In order to
connect those three copies, we add the edge $\{e^i_\vee, s^{i+1}_\vee \}$ for 
each $i\in [2]$. 
In the next step, we 
create three copies of the gadget $G_{=}$, denoted by $G^1_{=}$, $G^2_{=}$
and $G^3_{=}$, to simulate $a^1_1 \oplus a^2_1 =0$,
$a^1_2 \oplus a^2_2 =0$ and   $a^1_3 \oplus a^2_3 =0$. 
For each $i\in [3]$,
the vertex set of $G^{i}_=$ is defined by 
$\{s^i_=, e^i_= \}$. 
Again, we connect 
those three copies by adding $\{e^i_=, s^{i+1}_=\}$ for each $i\in [2]$
and we also create   $\{e^3_\vee, s^{1}_=  \}$ 
in order to connect  $G^3_{\vee}$ with $G^1_{=}$. The whole construction 
is illustrated in Figure~\ref{fig:9}.

\begin{figure}[h]
\begin{center}
\input{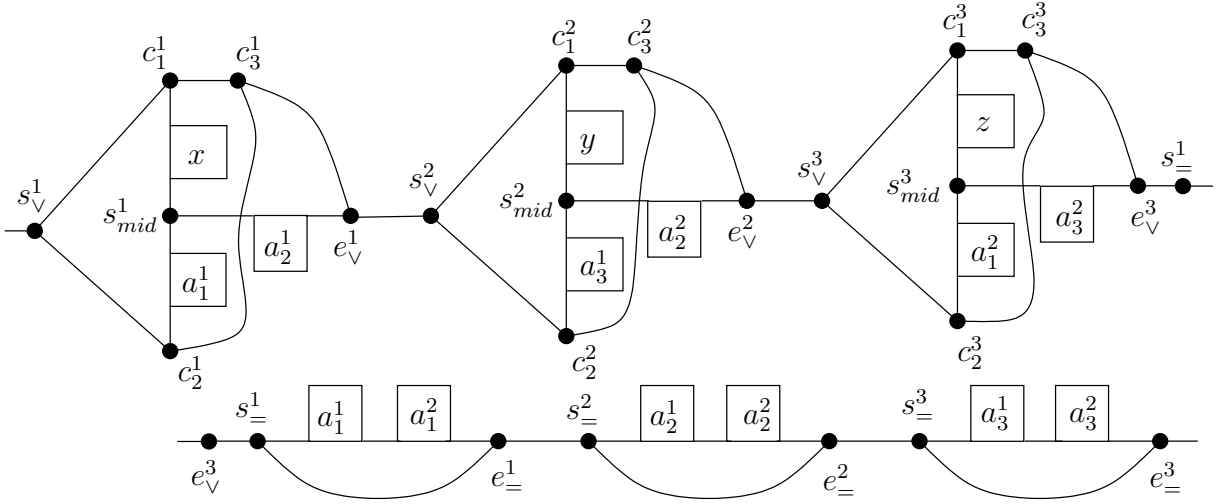}
\end{center}
\vspace{-0.5cm}
\caption{Modular view of the construction simulating  $x \oplus y \oplus z =0 $.}
\label{fig:9}
\end{figure} 

 Finally, we connect the  graphs that we introduced by 
parity gadgets as follows: For each graph $G^i_{=}$, we 
create two parity gadgets and connect
them to the   graph $G^{3j}_\vee$ corresponding the clause, in which the
 variable $a^k_i$ with $k\in \{1,2\}$
appear~(See Figure~\ref{fig:10} for a detailed view).  
The parity gadgets, which are associated to the variables $x$,
$y$ and $z$, are attached to   $G^{3j}_\vee$ with $j\in \{1,2,3 \}$
similarly as in the construction described in Section~\ref{sec:12oldconstr}  
for the graph $G^{3\oplus }_c$. Hence, the parity gadget is also connected to 
the graph  which is associated to the wheel $\cW_\alpha$, where
$ \alpha \in \{x,y,z\} $. 

Given an instance $I_\cH$ of the Hybrid problem, we refer to the corresponding
instance of the (1,2)-TSP in subcubic graphs as $G^{12}_{SC}$.

\subsection{Tour From Assignment}

We are going now to construct a tour from a given assignment and 
prove the following lemma.
\begin{lemma}
\label{lem:12sub:easy}
Let $I_\cH$ be an instance of the Hybrid problem
with $n$ wheels, $60\nu $ equation with two variables, 
$2\nu $ equations with three variables and $\phi$ an
assignment that leaves at most $\delta \nu $ equations unsatisfied.
Then, there is a tour in $G^{12}_{SC}$ with cost at most
  $672 \nu +3(n+1)-1+ \delta \nu $
\end{lemma}
\begin{proof}
Given the assignment $\phi$, we define the inner loop of the tour in $G^{12}_{SC}$
in the same way as in Lemma~\ref{lem:gen12:easy}. This means that some of the 
parity gadgets which are connected to gadgets simulating
equations with three variables may have been traversed in 
the inner loop of the tour.
In the outer loop of the tour, if the assignment satisfies the underlying
equation $x\oplus y \oplus z =0 $, then there is a Hamiltonian 
path traversing all graphs corresponding to $(x \vee a^{1}_1 \vee a^{1}_2)$,
$(y \vee a^{2}_2 \vee a^{1}_3)$, $(z \vee a^{2}_1 \vee a^{2}_3)$,
$a^1_1 \oplus a^2_1 =0$,
$a^1_2 \oplus a^2_2 =0$ and   $a^1_3 \oplus a^2_3 =0$.
For each satisfied equation with three variables, we  associate
the  cost $ 3(6 + 3\cdot 8+ 2   )$. If the underlying equation is not
satisfied, we have to introduce two endpoints. Thus, we associate the
cost $3(6 + 3\cdot 8+ 2   ) +1$. Summing up,  
 we obtain a tour in $G^{12}_{SC}$ with cost at most
$
8 \cdot 60\nu  + 3\cdot (6 + 3\cdot 8+ 2   )\cdot 2\nu + 3(n+1)-1  +\delta \nu = 
672 \nu +3(n+1)-1+ \delta \nu .
$ 
\end{proof}

\subsection{Assignment  From Tour}

Given a tour in $G^{12}_{SC}$, we are going to construct 
an assignment to the variables of the corresponding instance
$I_\cH$ of the Hybrid problem and prove the following lemma.

\begin{lemma}
\label{lem:12sub:hard}
Let $I_\cH$ be an instance of the Hybrid problem
with $n$ wheels, $60\nu $ equations with two variables, 
$2\nu $ equations with three variables and $\pi$ a
tour in $G^{12}_{SC}$ with cost  $ 672  \nu +3(n+1)-1 + \delta \nu $.
Then, it is possible to construct efficiently 
an assignment that leaves at most $\delta \nu $ equations in 
$I_\cH$ unsatisfied.
\end{lemma}
\begin{proof}

In the first step, we convert the underlying tour in $G^{12}_{SC}$ into a consistent 
one without increasing its cost. This is done by applying
Lemma~\ref{lem:12gen:cons}  to each parity gadget in $G^{12}_{SC}$.
In the second step, we use the same $0/1$-traversals of the parity gadgets 
 in the inner loop of the tour which enables us to 
construct a tour 
 in the corresponding instance  $G^{12}_{\cH}$ with cost at most 
$$
 672 \nu +3(n +1) -1 + \delta \nu    - 3\cdot (6 + 3\cdot 8+ 2   ) \cdot 2\nu +    
(3\cdot 8 +3) \cdot 2\nu = 534\nu + 3(n +1) -1 + \delta \nu.
 $$
Finally, we apply Lemma~\ref{lem:12gen:hard} and compute efficiently an assignment that 
leaves at most $\delta \nu $ equations in  $I_\cH$ unsatisfied.
\end{proof}

We are ready to give the proof of Theorem~\ref{thm:12subcubic}.

\begin{proof}[Proof of Theorem~\ref{thm:12subcubic}]
Given $I_\cH$ an  instance of the Hybrid problem consisting of $n$ wheels,
 $60 \nu $ equations with two variables
and $2\nu $ equations with three variables,
 we construct in polynomial time 
the associated instance $G^{12}_{SC}$ of the
(1,2)-TSP.

Given an assignment $\phi$ to the variables of $I_\cH$ leaving $\delta \cdot \nu$ equations 
unsatisfied with $\delta \in (0, 1)$, 
then, according to Lemma~\ref{lem:12sub:easy}, it is possible to find a tour 
with  cost at most $672 \nu +3(n+1)-1+ \delta\cdot  \nu$.

On the other hand, if we are given a tour $\sigma$ in $G^{12}_{SC}$ with cost 
$672 \nu +3(n+1)-1+ \delta\cdot  \nu $, due to Lemma~\ref{lem:12sub:hard}, 
we are able to  construct efficiently an assignment to the variables 
of $I_\cH$, which leaves
at most $\delta \nu $ equations in $I_\cH$ unsatisfied. 

Similarly to the 
proof of Theorem~\ref{thm:12gen:ks12}, for a constant $\tau >0$,  we may assume that $(3n + 4)/ \nu \leq \tau$
holds. According to Theorem~\ref{ssphybridsatz}, we know that for all $\eps> 0$, it 
is \clNP-hard to decide 
 whether there is a tour with cost
at most $ 672 \nu+ 3(n + 1)-1+ \eps \cdot \nu
 \leq 672\cdot \nu +\eps'\nu$
or all tours have cost at least $672\cdot \nu + (1 -\eps) \nu+3(n + 1)-1  \geq 673 \cdot 
\nu
-\eps'\cdot \nu  $, for some $\eps'$ that depends only on $\eps$ and $\tau$.
By appropriate choices for $\epsilon$ and  $\tau$, 
the ratio between these two cases can get arbitrarily close to $673/672$. 
\end{proof}

\begin{figure}[h]
\begin{center}
\input{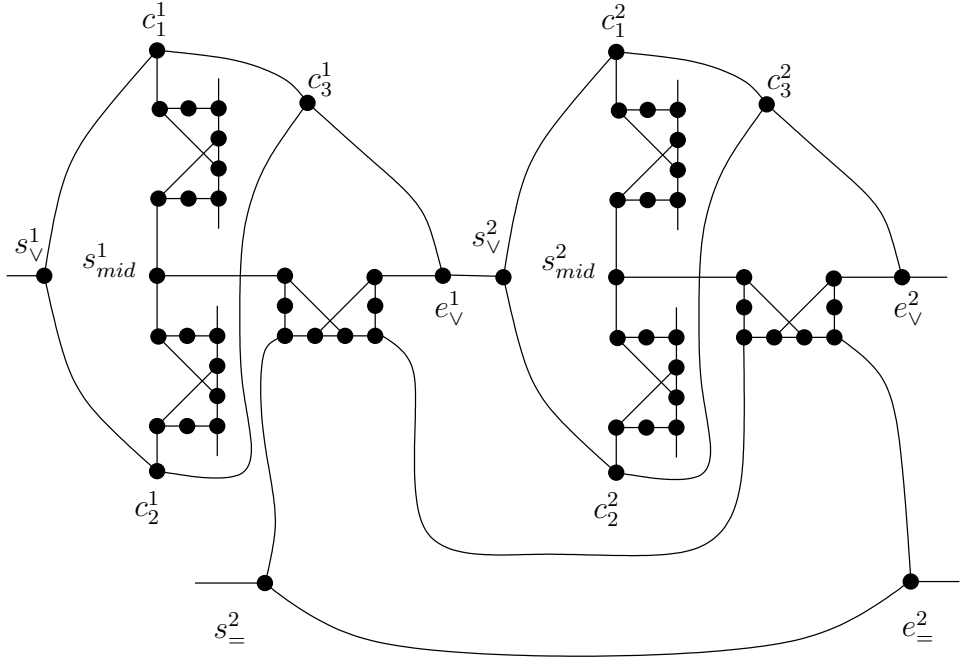}
\end{center}
\vspace{-0.5cm}
\caption{Detailed view of the gadget for  $(x \vee a^{1}_1 \vee a^{1}_2)$,
$(y \vee a^{2}_2 \vee a^{1}_3)$ and $a^1_2 \oplus a^2_2 =0$.}
\label{fig:10}
\end{figure} 

\section{(1,2)-TSP in Cubic Graphs}
\label{sec:12tspcub}
This section is devoted to  
 the proof of Theorem~\ref{thm:12cubic}.

\subsection{The Construction of the Graph $G^{12}_{CU}$} 
Given an instance $I_\cH$ of the Hybrid problem with $n$ wheels, $60 \nu$
equations with two variables and $2\nu $ equations with three variables,
we construct the corresponding graph $G^{12}_{SC}$. In order to 
convert the instance $G^{12}_{SC}$ of the (1,2)-TSP in subcubic graphs
into an instance $G^{12}_{CU}$ of the (1,2)-TSP in cubic graphs, 
we   replace all vertices with degree exactly two by 
a path in which all vertices will have degree exactly three. Let us 
describe this in detail: 
Let $w$ be a vertex with degree two in $G^{12}_{SC}$,  
which is connected to $x$ and
$y$. Replace  $w$ with  the path
$p_w=v^1_w-v^2_w-v^3_w-v^4_w $. In addition, we add edges $\{v^1_w, v^3_w\}$, 
 $\{v^2_w, v^4_w \}$, $\{x, v^1_w\}$ and $\{y , v^4_w\}$. 
By applying this modification to each vertex of degree exactly two,
we create a cubic graph  and refer to it as   $G^{12}_{CU}$.

A modified parity gadget is  displayed in Figure~\ref{fig:2}~$(a)$.
The corresponding traversals are defined in Figure~\ref{fig:2}~$(b)$ and $(c)$.

\begin{figure}[h]
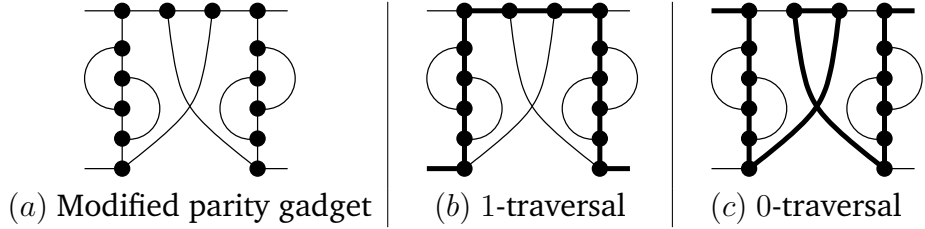

\begin{center}
\begin{tabular}[c]{c|c|c}
\input{paritz2.pdf_t} & ~ \input{paritz20.pdf_t} ~ & 
~ \input{paritz21.pdf_t} \\
$(a)$~Modified parity gadget  & $(b)$~$1$-traversal &  $(c)$~$0$-traversal 
\end{tabular}
\vspace{-0.5cm}
\end{center}
\caption{$0/1$-Traversals of a modified parity gadget. 
The  traversed edges are pictured by thick lines.}
\label{fig:2}
\end{figure}

The following lemma enables us to construct a tour in $G^{12}_{CU}$
given an assignment  $\phi$ to the variables of the corresponding instance $I_\cH$  of the Hybrid problem
with a certain cost that depends on the number on unsatisfied equations 
in $I_\cH$ by $\phi$.

\begin{lemma}
\label{lem:12cub:easy}
Let $I_\cH$ be an instance of the Hybrid problem
with $n$ wheels, $60\nu $ equation with two variables, 
$2\nu $ equations with three variables and $\phi$ an
assignment that leaves  $\delta\cdot  \nu$ equations unsatisfied
for some $\delta \in (0,1)$.
Then, it is possible to construct efficiently  
a tour in $G^{12}_{CU}$ with cost at most
  $1140\nu + 6(n+1)-1  +  \delta\cdot  \nu$
\end{lemma}
\begin{proof}
Basically, we  use the same tour as constructed in Lemma~\ref{lem:12sub:easy} for 
the graph $G^{12}_{SC}$ with 
the difference that instead of traversing a vertex $w$ of degree exactly
two in $G^{12}_{SC}$, we have to use the path $v^1_w-v^2_w-v^3_w-v^4_w$ 
consisting of $3$ more vertices. Thus, if we have given a tour $\sigma$ in
$G^{12}_{SC}$, that was constructed according to Lemma~\ref{lem:12sub:easy}, we have to add
 $ 6\cdot  60\nu $ (for each equation with two variables),  $9\cdot 6 \cdot 2\nu $
(for each equation with three variables), and 
$ 3(n+1) $ (for each wheel) to the cost of $\sigma$ and obtain a tour in $G^{12}_{CU}$
 with cost at most 
$$
672 \nu + 3(n +1) -1 + \delta\cdot  \nu +  (6\cdot  60\nu ) +  9\cdot 6 \cdot 2\nu 
+ 3(n+1)= 1140\nu + 6(n+1)-1 +\delta\cdot  \nu
  $$
	and the proof of Lemma~\ref{lem:12cub:easy} follows.
\end{proof}

\subsection{Tour to Assignment}
We are going to prove the other direction of the reduction and
give the  proof of the following lemma.
\begin{lemma}
\label{lem:12cubhard}
Let $I_\cH$ be an instance of the Hybrid problem
with $n$ wheels, $60\nu $ equation with two variables, 
$2\nu $ equations with three variables and $\pi$ a
tour  in $G^{12}_{CU}$ with cost  $ 1140 \nu +  6(n+1) -1+ \delta \cdot \nu $.
Then, it is possible to construct efficiently 
an assignment that leaves at most $\delta \cdot \nu$ equations in 
$I_\cH$ unsatisfied.
\end{lemma}


%
%

\begin{proof}
Let $\pi$ be a tour in $G^{12}_{CU}$ with cost $ 1140 \nu +  6(n+1)-1 + \delta \cdot \nu $.
We are going to show that we can convert efficiently 
$\pi $  into a tour $\pi'$ in $G^{12}_{SC}$ with cost
$ 672 \nu +3(n+1)-1+ \delta\cdot  \nu $. For this, we consider 
  the path $x- v^1_c-v^2_c-v^3_c-v^4_c - y $ in 
$G^{12}_{CU}$, where $p_c= v^1_c-v^2_c-v^3_c-v^4_c $ corresponds to the  vertex $c$ of 
degree exactly two in the instance $G^{12}_{SC}$. 
As we want to contract the path $p_c$ into one vertex, we will ensure that 
the (1,2)-tour   is using either the path $ v^1_c-v^2_c-v^3_c-v^4_c $ or 
$ v^1_c - v^3_c - v^2_c - v^4_c $. Let us assume that either $v^2_c$ or $v^3_c$
is an endpoint, say $v^2_c$. Clearly, it implies that there is another endpoint in 
$\{ v^1_c,v^3_c,v^4_c\}$ or $v^2_c$ is a double endpoint.
We delete all edges of weight 1 that the tour is using and are
 incident on $v^2_c$ and $v^3_c$. Then, we add $\{v^1_c,v^2_c\}$, $\{v^2_c,v^3_c\}$
and   $\{v^3_c, v^4_c\}$ to connect $v^4_c$ and $v^1_c$ by edges of weight 1.
Note that this transformation decreased the total number of endpoints and the cost of 
the (1,2)-tour.  By applying this
transformation successfully to each such path $p_c$, we obtain a tour which 
is using the complete   
path that corresponds to a vertex of degree 2 in the instance $G^{12}_{SC}$
without increasing the cost of the tour. By contracting 
 each path $p_c$ into  the vertex $c$, it yields  a (1,2)-tour
in $G^{12}_{SC}$ with cost at most $672 \nu +3(n+1)+1+ \delta\cdot  \nu$. 
Finally, we apply lemma~\ref{lem:12sub:hard} and obtain an assignment 
that leaves at most $\delta \cdot \nu $ equations in $I_\cH$ unsatisfied.
\end{proof}
Analogously to the proof of Theorem~\ref{thm:12subcubic}, we
combine Lemma~\ref{lem:12cub:easy} with Lemma~\ref{lem:12cubhard} and obtain  
Theorem~\ref{thm:12cubic}. 

\section{Graphic TSP in Subcubic and Cubic Graphs}
\label{sec:graphtsp}
In this section, we are going to give the proof of Theorem~\ref{thm:gr:cubic}
and Theorem~\ref{thm:gr:subcubic}.

\subsection{The Construction}
Let $I_\cH$ be an instance of the Hybrid problem. We first construct 
the corresponding instances  $G^{12}_{CU}$ and $G^{12}_{SC}$
 of the (1,2)-TSP in cubic and subcubic graphs, respectively.
 Each gadget $G_=$ in  $G^{12}_{SC}$
is replaced by the graph $G^{gr}_=$ displayed in Figure~\ref{fig:5}.
We refer to this construction as the graph $G^{gr}_{SC}$.
In order to obtain  an instance of the  Graphic TSP on 
cubic graphs, we use the modified parity gadgets in $G^{gr}_=$
and denote this instance as  $G^{gr}_{CU}$.

\begin{figure}[h]
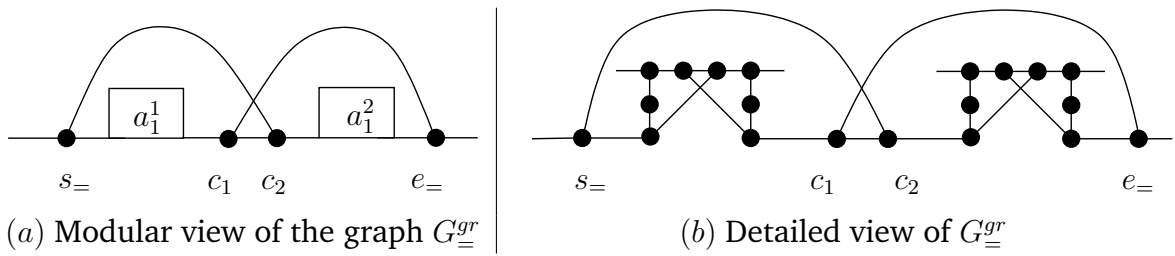

\begin{center}
\begin{tabular}[c]{c|c}
\input{equalgr.pdf_t} & ~ \input{equalgrdet.pdf_t}  \\
$(a)$~Modular view of the graph $G^{gr}_{=}$  & $(b)$~Detailed view of $G^{gr}_{=}$ 
\end{tabular}
\end{center}
\vspace{-0.5cm}
\caption{Graph $G^{gr}_=$ corresponding to $a^1_1 \oplus a^2_1 =0 $.}
\label{fig:5}
\end{figure} 

Let us prove one direction of the reductions.
\begin{lemma}
\label{lem:grsub:easy}
Let $I_\cH$ be an instance of the Hybrid problem
with $n$ wheels, $60\nu $ equation with two variables, 
$2\nu $ equations with three variables and $\phi$ an
assignment that leaves at most $\delta \nu $ equations unsatisfied.
Then, there is a tour in $G^{gr}_{SC}$ and in $G^{gr}_{CU}$
 with cost at most $684  \nu +3(n+1)-1 + \delta \nu $ and $1152\nu+6(n+1)-1+ \delta 
\nu $, respectively.
\end{lemma}

\begin{proof}
Let us start with the description of the tour in $G^{gr}_{SC}$.
As for the inner loop, 
we use the same tour as in Lemma~\ref{lem:12sub:easy}.
Note  that we traversed only edges with weight 1 in the inner 
loop of the tour in $G^{12}_{SC}$.
In the outer loop, we cannot use the same shortcuts as  in the
(1,2)-TSP, since  in some cases the weight of  an edge can be 
greater than 2. To 
ensure that the cost traversing a gadget corresponding
to an equation with three variables increases only by 
one if the equation is unsatisfied by the assignment,
we will use the following trick:
Consider an equation of the form $x\oplus y \oplus z =0$ 
that is simulated by 
$(x \vee a^{1}_1 \vee a^{1}_2)$,
$(y \vee a^{2}_2 \vee a^{1}_3)$, $(z \vee a^{2}_1 \vee a^{2}_3)$,
$a^1_1 \oplus a^2_1 =0$,
$a^1_2 \oplus a^2_2 =0$ and   $a^1_3 \oplus a^2_3 =0$. If we have an assignment
that satisfies $x\oplus y \oplus z =0$, then there is also an assignment that 
satisfies all $6$ associated predicates. Furthermore, we see that in the other case,
we can find an assignment that satisfies all predicates except exactly one equation  
 with two variables. 

In particular, it implies 
for a tour traversing the gadget $G^{gr}_=$ simulating  $a^1_1 \oplus a^1_2 =0$
 that if $(a^1_1 + a^1_2 =0)$ and  $(a^1_1 + a^1_2 =2)$ holds, 
we use $s_= - c_2 - c_1 - e_=   $ 
and $s_= - c_2 - c_1 - e_=   $, respectively. On the other hand, 
assuming $(a^1_1 + a^1_2 =1)$, we traverse 
either $s_= -c_1 -  c_2 - c_1 - e_=   $ or $s_= -c_2 -  c_1 - c_2 - e_=   $.
Thus, we use  the edge $\{c_1, c_2\}$ twice increasing the cost only by 1.  

Summarizing, given an assignment leaving $\delta \nu  $ equations unsatisfied, we find
a tour in $G^{12}_{SC}$ with cost at most  $672  \nu +3(n+1)-1 + \delta \nu$ and a tour 
in  $G^{gr}_{SC}$ with cost at most $684  \nu +3(n+1)-1 + \delta \nu $, 
since we have to take 
into account the small detour and 
 add $3\cdot 2 \cdot 2\nu  $ to the cost.

Under the same conditions, 
we find a tour in $G^{12}_{CU}$ with cost at most $1140  \nu +6(n+1)-1 + \delta \nu $ and a tour 
in  $G^{gr}_{CU}$ with cost at most $1152  \nu +6(n+1)-1 + \delta \nu $.
\end{proof}

\subsection{Tour to Assignment}
We now give the other direction of the reductions and prove the following lemma.
\begin{lemma}
\label{lem:grsub:hard}
Let $I_\cH$ be an instance of the Hybrid problem
with $n$ wheels, $60\nu $ equation with two variables, 
$2\nu $ equations with three variables,  $\pi$ a
tour in $G^{gr}_{SC}$ with cost  $ 684  \nu +3(n+1)-1 + \delta \nu $
and $\sigma$ a
tour in $G^{gr}_{CU}$ with cost  $ 1152  \nu +6(n+1)-1 + \delta \nu $.
By using either $\pi$ or $\sigma$, 
it is possible to construct efficiently 
an assignment that leaves at most $\delta \nu $ equations in 
$I_\cH$ unsatisfied.
\end{lemma}
\begin{proof}
Let us  consider a tour
$\pi$  in $G^{gr}_{SC}$ with cost $ 684  \nu +3(n+1)-1 + \delta \nu $.
We  interpret  $\pi$ as  a (1,2)-tour in $G^{gr}_{SC}$ with cost 
at most $ 684  \nu +3(n+1)-1 + \delta \nu $.   
In the first step, we convert the underlying tour in $G^{gr}_{SC}$ into a consistent 
one without increasing its cost  by applying
Lemma~\ref{lem:12gen:cons}  to each parity gadget in $G^{gr}_{SC}$.
In the second step, we use the same $0/1$-traversals of the parity gadgets 
 in the inner loop  which enables us to 
construct a tour 
 in the corresponding instance  $G^{12}_{SC}$ with cost at most 
$ 672  \nu +3(n+1)-1 + \delta \nu  $.
Finally, we apply Lemma~\ref{lem:12sub:hard} and construct an assignment 
leaving at most $\delta  \nu $ equations in $I_\cH$ unsatisfied.

Analogously, if we have given a tour in $G^{gr}_{SC}$
with cost  $ 1152  \nu +6(n+1)-1 + \delta \nu $, we convert it 
into a (1,2)-tour without increasing its cost. By applying the contractions defined
in Lemma~\ref{lem:12cubhard},
we obtain a (1,2)-tour in  $G^{gr}_{SC}$ with cost 
at most $ 684  \nu +3(n+1)-1 + \delta \nu $, for which we already know 
how to construct an assignment with the desired properties.
\end{proof}

By combining Lemma~\ref{lem:grsub:easy} and Lemma~\ref{lem:grsub:hard}, we obtain immediately Theorem~\ref{thm:gr:cubic}
and Theorem~\ref{thm:gr:subcubic}.

\section{Summary of the Inapproximability Results}
\label{sec:summary}
As mentioned before the explicit inapproximability bound of 535/534 
(\cite{KS12},\cite{KS13}) for the (1,2)-TSP carries through to the Graphic TSP.
We summarize here (Table~\ref{fig:inapproxresults}) the results of the paper.

\begin{table}[h]
\begin{center}
\begin{tabular}{|c|c|c|}
\hline
 & & \\
\raisebox{1.5ex}{
\textbf{Restriction} }& 
\raisebox{1.5ex}{
\textbf{(1,2)-TSP}} & 
   \raisebox{1.5ex}{ \textbf{Graphical TSP } }\\
\hline
\hline
&&\\
\raisebox{1.5ex}{Unrestricted}  &  \raisebox{1.5ex}{535/534 }       &  
\raisebox{1.5ex}{535/534}      \\
\hline
&& \\
 \raisebox{1.5ex}{Subcubic } & \raisebox{1.5ex}{673/672    }      &     
\raisebox{1.5ex}{ 685/684 }\\
 \hline
&& \\
\raisebox{1.5ex}{Cubic}  &  \raisebox{1.5ex}{ 1141/1140 }       & 
\raisebox{1.5ex}{ 1153/1152  }   \\
\hline
\end{tabular}
\end{center}
\caption{
Inapproximability bounds for the instances of (1,2)-TSP and Graphic TSP.
}
\label{fig:inapproxresults}
\end{table}



\section{Conclusions and Further Research } 
\label{sec:conclusion}
We provided new explicit inapproximability bounds for 
cubic and subcubic instances of (1,2)-TSP and Graphic
TSP. The important question is to improve the explicit 
inapproximability bounds on those instances significantly.
A  bottleneck in our constructions, especially for  the cubic case, are the parity
gadgets. Using the modularity of the constructions, any improvement of the costs
of the parity gadgets
will lead to improved inapproximability bounds for the corresponding problems.

The current best upper approximation bound for general cubic instances
of Graphic TSP is  4/3 (cf. \cite{BSSS11a}). 
For the special case of 2-connected cubic graphs, the bound was recently 
improved to (4/3 -  1/61236) \cite{CLS12}.   
How about further improving
 those bounds? How about improving the 
general upper bound of 8/7 \cite{BK06} for cubic instances
of the (1,2)-TSP?


\section*{Acknowledgments}
We thank Leen Stougie and Ola Svensson for a number of interesting
discussions.

\end{document}